\documentclass[1p]{elsarticle}
\usepackage{hyperref}
\journal{}
\usepackage{fix-cm}
\usepackage{mathtools}
\usepackage{amsmath,amsthm}
\usepackage{graphicx}
\usepackage{latexsym,enumerate,amssymb,amsbsy}
\usepackage{graphics,color}
\usepackage{verbatim}
\usepackage{url}
\usepackage{mathptmx}
\usepackage{pgf,tikz}
\usetikzlibrary{trees,automata,arrows,shapes.geometric}
\usepackage{amsfonts}
\usepackage{stmaryrd}
\usepackage{algorithm}
\usepackage[noend]{algpseudocode}
\usepackage{url}
\usepackage{colortbl}
\newtheorem{proposition}{Proposition}
\newtheorem{theorem}[proposition]{Theorem}
\newcommand{\proofofref}{}
\newproof{zproofof}{Proof of \proofofref}

\newtheorem{lemma}[proposition]{Lemma}

\newtheorem{example}{Example}

\newtheorem{remark}{Remark}

\newcommand{\Mod}[1]{\ (\textup{mod}\ #1)}

\def\CC{{\mathcal{C}}}
\def\CP{{\mathcal{P}}}

\DeclareMathOperator{\lcm}{lcm}

\usepackage{listings}
\usepackage{xcolor}

\definecolor{codegreen}{rgb}{0,0.6,0}
\definecolor{codegray}{rgb}{0.5,0.5,0.5}
\definecolor{codepurple}{rgb}{0.58,0,0.82}
\definecolor{backcolour}{rgb}{0.95,0.95,0.92}

\lstdefinestyle{mystyle}{
	backgroundcolor=\color{backcolour},
	commentstyle=\color{codegreen},
	keywordstyle=\color{magenta},
	numberstyle=\tiny\color{codegray},
	stringstyle=\color{codepurple},
	basicstyle=\ttfamily\footnotesize,
	breakatwhitespace=false,
	breaklines=true,
	captionpos=b,
	keepspaces=true,
	numbers=left,
	numbersep=5pt,
	showspaces=false,
	showstringspaces=false,
	showtabs=false,
	tabsize=2
}

\algdef{SE}[DOWHILE]{Do}{doWhile}{\algorithmicdo}[1]{\algorithmicwhile\ #1}

\begin{document}
\begin{frontmatter}
\title{An Efficiently Generated Family of Binary de Bruijn Sequences}

\author[zzu]{Yunlong Zhu}
\ead{zhu2010\_546@163.com}

\cortext[cor1]{Corresponding author}
\author[zzu]{Zuling Chang\corref{cor1}}
\ead{zuling\_chang@zzu.edu.cn}
		
\author[ntu]{Martianus Frederic Ezerman}
\ead{fredezerman@ntu.edu.sg}
		
\author[wang]{Qiang Wang}
\ead{wang@math.carleton.ca}
		
\address[zzu]{School of Mathematics and Statistics, Zhengzhou University,
450001 Zhengzhou, China.}
		
\address[ntu]{School of Physical and Mathematical Sciences, Nanyang Technological University,\\
21 Nanyang Link, Singapore 637371.}

\address[wang]{School of Mathematics and Statistics, Carleton University,\\
1125 Colonel By Drive, Ottawa, Ontario, K1S 5B6, Canada.}
		
\begin{abstract}
We study how to generate binary de Bruijn sequences efficiently from the class of simple linear feedback shift registers with feedback function $f(x_0, x_1, \ldots, x_{n-1}) = x_0 + x_1 + x_{n-1}$ for $n \geq 3$, using the cycle joining method. Based on the properties of this class of LFSRs, we propose two new generic successor rules, each of which produces at least $2^{n-3}$ de Bruijn sequences. These two classes build upon a framework proposed by Gabric, Sawada, Williams and Wong in \emph{Discrete Mathematics} vol.~341, no.~11, pp. 2977--2987, November 2018. Here we introduce new useful choices for the uniquely determined state in each cycle to devise valid successor rules. These choices significantly increase the number of de Bruijn sequences that can be generated. In each class, the next bit costs $O(n)$ time and $O(n)$ space for a fixed $n$.
\end{abstract}
		
		
\end{frontmatter}
	

\section{Introduction}\label{sec:intro}

A binary {\it de Bruijn sequence} of order $n$ is a $2^n$-periodic sequence in which each $n$-tuple occurs exactly once per period. There are $2^{2^{n-1}-n}$ such sequences~\cite{Bruijn46}. They have been studied for a long time as they appeared in multiple disguises~\cite{Ral82}. More details are supplied in Fredricksen's survey~\cite{Fred82}. Certain families of such sequences have been found useful in far ranging application domains that include bioinformatics, communication systems, coding theory, and cryptography.

One can build de Bruijn sequences of order $n$ from the Hamiltonian paths of an $n$-dimensional de Bruijn graph over $2$ symbols. This is equivalent to finding Eulerian cycles of an $(n-1)$-dimensional de Bruijn graph. While a complete enumeration of all such cycles can be done, for example by Fleury's algorithm \cite{Fleury}, this rather naive approach is highly inefficient in storage requirement. It remains a major objective to strike a good balance between minimizing the computational costs and maximizing the number of sequences that can be explicitly built. On top of this consideration, depending on the specific application domains, additional requirements may be imposed. In cryptography, for instance, the preference is towards de Bruijn sequences with particular linear complexity profiles while in DNA fragment assembly certain substrings may be more or less desirable than others.

A well-known generic construction approach is the {\it cycle joining  method} (CJM) (see, for examples, \cite{Fred82} and \cite{Golomb}). The main idea of this method is to join all cycles produced from a given Feedback Shift Register (FSR) into a single cycle by interchanging the successors of some pairs of conjugate states. A good number of CJM-based fast algorithms are already in the literature. Most of them produce a very limited number of sequences. Let us sample a few. As was shown in~\cite{Fred72}, one can generate the {\tt granddaddy} de Bruijn sequence in $O(n)$ time and $O(n)$ space per bit. A related sequence, the {\tt grandmama}, was built in~\cite{Dragon18}. Etzion and Lempel proposed some algorithms to generate de Bruijn sequences based on the {\it pure cycling register} (PCR) and the {\it pure summing register} (PSR) in~\cite{Etzion84}. Their algorithms generate a remarkable number, an exponential of $n$, of sequences at the expense of higher memory requirement to store some special states. Earlier, Huang gave another construction that joins the cycles of the {\it complemented cycling register} (CCR) in~\cite{Huang90}. Jansen, Franx, and Boekee established a requirement to determine some conjugate pairs in \cite{Jansen91}, leading to another fast algorithm. In \cite{Sawada16}, Sawada, Williams, and Wong proposed a simple de Bruijn sequence construction, which turned out to be a special case of the method in~\cite{Jansen91}. Gabric, Sawada, Williams, and Wong generalized the last two results to form a simple successor rule framework that yield three new and simple  de Bruijn constructions based on the PCR and the CCR in~\cite{Gabric18}. Further generalization to the constructions of $k$-ary de Bruijn sequences was done in~\cite{Sawada17} and \cite{Gabric19}. Chang, Ezerman, Ke, and Wang recently proposed a new criteria for successor rules in~\cite{chang19}. They applied the criteria to efficiently construct numerous de Bruijn sequences based on the PCR and the PSR by imposing new relations on the respective generated cycles.

In this paper we provide more successor rules to generate a new family of de Bruijn sequences. We use the CJM to join all cycles generated by a special LFSR whose characteristic polynomial is $f(x)=x^n+x^{n-1}+x+1 \in \mathbb{F}_2[x]$ for $n \geq 3$. The cycles generated by this LFSR correspond to the cycles of the PCR and CCR of order $n-1$. Sala, in a Master's thesis \cite{SalaTh}, studied this LFSR and proposed a successor rule to generate de Bruijn sequences in $O(n)$ time and $O(n)$ space per bit. In a recent preprint \cite{Salasame}, Sala, Sawada, and Alhakim noticed that the states in each cycle produced by this LFSR have the same run-length. They  then proposed a new successor rule based on the so-called {\it run-length order} to generate the famous {\tt  prefer-same} de Bruijn sequence \cite{Eldert58} using $O(n)$ time per bit and only $O(n)$ space. They named the LFSR the {\it pure run-length register} (PRR). This, along with the general framework to generate de Bruijn sequences proposed in \cite{Gabric18} and \cite{Salasame}, can be found in \cite{Sawada-web}.

Our main contribution in this paper is to construct two new generic successor rules based on the special LFSR by applying the main results of \cite{Gabric18}. The number of de Bruijn sequences generated by our new rules is exponential in the order $n$ of the FSR. Our method runs in $O(n)$ time and $O(n)$ space per generated bit. More explicitly, we accomplish the following tasks. 

\begin{enumerate}
\item We take a different point of view in studying the PRR of order $n \geq 3$ from the one already done in~\cite{SalaTh}. Here we discuss the properties of the PRR of order $n$ via the characteristic polynomials of the PCR and the CCR of order $n-1$.

\item Two new generic successor rules are presented based on the PRR to generate de Bruijn sequences. The correctness of the rules are demonstrated by applying the framework proposed in \cite{Gabric18}. We introduce the notion of \emph{a critical set of spanning conjugate pairs} to effectively define a spanning tree of the cycles induced by the PRR. In each generic rule, a CCR-related cycle has a unique parent  while a PCR-related cycle may have a variety of possible parents, depending on the specified critical set.

\item For each generic successor rule, we construct critical sets to uniquely identify the respective parents of the PCR-related cycles efficiently in both time and space. Thus, each rule leads to an exponential number of de Bruijn sequences that can be generated from among the generic successors. Each of the formulated algorithms can generate the next bit of a de Bruijn sequence in $O(n)$ time using $O(n)$ space.
\end{enumerate}

In terms of organization, Section~\ref{sec:prelim} gathers some preliminary notions and useful known results. Sections~\ref{sec:class1} and~\ref{sec:class2} provide the treatment on the two classes, respectively. We end with a brief discussion on the required computational resources, a few directions for follow-up investigation, and the \texttt{C} code of our basic implementation.

\section{Preliminaries}\label{sec:prelim}

An {\it $n$-stage shift register} is a clock-regulated circuit with the following properties. It has $n$ consecutive storage units. Each unit holds a bit. As the clock pulses the circuit shifts the bit in each unit to the next stage. The register turns into a binary code generator if one appends a feedback loop that outputs a new bit $s_n$ based on the $n$-stage {\it initial state} ${\bf s}_0= s_0,\ldots,s_{n-1}$. The corresponding Boolean {\it feedback function} $f(x_0,\ldots,x_{n-1})$ outputs $s_n$ on input ${\bf s}_0$. A {\it feedback shift register} (FSR), therefore, outputs a binary sequence ${\bf s}=\{s_i\}=s_0,s_1,\ldots,s_n,$ $\ldots$ that satisfies the recursive relation
\[
s_{n+\ell} = f(s_{\ell},s_{\ell+1},\ldots,s_{\ell+n-1}) \text{ for } \ell = 0,1,2,\ldots.
\]
For $N \in \mathbb{N}$, if $s_{i+N}=s_i$ for all $i \geq 0$, then ${\bf s}$ is {\it $N$-periodic} or {\it with period $N$} and we write ${\bf s}= (s_0,s_1,s_2,\ldots,s_{N-1})$. The least among all periods of ${\bf s}$ is called the {\it least period} of ${\bf s}$.

We call ${\bf s}_i= s_i,s_{i+1},\ldots,s_{i+n-1}$ {\it the $i$-th state} of ${\bf s}$. The {\it predecessor} and the {\it successor} of ${\bf s}_i$ are denoted, respectively, by ${\bf s}_{i-1}$ and ${\bf s}_{i+1}$. For $s \in \mathbb{F}_2$, let $\bar{s} :=s+1 \in \mathbb{F}_2$. The definition extends to any binary vector or sequence. If ${\bf s}= s_0,s_1,\ldots,s_{n-1}, \ldots$, then $\bar{\bf s} :=
\bar{s}_0,\bar{s}_1,\ldots,\bar{s}_{n-1}, \ldots$. For an arbitrary state ${\bf v}=v_0,v_1,\ldots,v_{n-1}$ of ${\bf s}$, the states
\[
\hat{{\bf v}}:=\bar{v}_0,v_1,\ldots,v_{n-1} \mbox{ and }
\tilde{{\bf v}}:=v_0,\ldots,v_{n-2},\bar{v}_{n-1}
\]
are the {\it conjugate} state and {\it companion} state of ${\bf v}$, respectively. Hence, $({\bf v}, \hat{{\bf v}})$ is a {\it conjugate pair} and $({\bf v}, \tilde{{\bf v}})$ is a {\it companion pair}.

Any FSR with feedback function $f$, on distinct $n$-stage initial states, generates distinct sequences that form a set $G(f)$ of cardinality $2^n$. All sequences in $G(f)$ are periodic if and only if the feedback function $f$ is {\it nonsingular}, that is, $f(x_0,x_1,\ldots,x_{n-1}) = x_0+h(x_1,\ldots,x_{n-1})$ for some Boolean function $h(x_1,\ldots,x_{n-1})$ whose domain is $\mathbb{F}_2^{n-1}$~\cite[p.~116]{Golomb}. Here we deal only with nonsingular feedback functions. An FSR is {\it linear} or an LFSR if its feedback function is
$f(x_0,\ldots,x_{n-1}) = x_0+a_1x_1 + \cdots + a_{n-1} x_{n-1}$. The polynomial
\[
f(x)=x^n + a_{n-1} x^{n-1} +\cdots  + a_1 x + 1\in\mathbb{F}_2[x]
\]
is the {\it characteristic polynomial} of the LFSR. Otherwise, the FSR is {\it nonlinear} or an NLFSR. Further properties of LFSRs are treated in \cite{GG05} and \cite{LN97}.

The {\it left shift operator} $L$ maps a periodic sequence
\[
\mathbf{s}:=(s_0,s_1,\ldots,s_{N-1}) \mapsto L \mathbf{s}:=(s_1,\ldots,s_{N-1},s_0),
\]
with the convention that $L^0$ fixes $\mathbf{s}$. The set
\[
[{\bf s}]:=\left\{{\bf s},L{\bf s},\ldots,L^{N-1}{\bf s} \right\}
\]
is a {\it shift equivalent class}. Sequences in the same shift equivalent class correspond to the same cycle in the state diagram of the FSR \cite{GG05}. A {\it cycle} is a periodic sequence in a shift equivalent class. If an FSR with feedback function $f$ generates exactly $r$ distinct cycles $C_1, C_2, \ldots, C_r$, then its {\it cycle structure} is
\[
\Omega(f)=\{C_1, C_2, \ldots, C_r\}.
\]
A cycle can also be viewed as a set of $n$-stage states in the corresponding periodic sequence. When $r=1$, the corresponding FSR is of {\it maximal length} and its output is a de Bruijn sequence of order $n$.

The unique lexicographically least $n$-stage state in each cycle $C \in \Omega(f)$ is designated as the {\it cycle representative} of $C$ in \cite{Jansen91}. One can impose a lexicographic order $\prec_{\rm lex}$ on the cycles, based on their representatives,  by saying that $C_i { \prec}_{\rm lex} C_j$ if and only if the cycle representative of $C_i$ is lexicographically smaller than that of $C_j$.

If two distinct cycles $C_i$ and $C_j$ in $\Omega(f)$ have the property that
the state  ${\bf v}=v_0,v_1,\ldots,v_{n-1} \in C_i$ has its conjugate state $\hat{{\bf v}} \in C_j$, then interchanging the successors of ${\bf v}$ and $\hat{{\bf v}}$ joins $C_i$ and $C_j$ into a single cycle. The feedback function of this new cycle is
\begin{equation}\label{eq:newfeedback}
\hat{f}:=f(x_0,x_1,\ldots,x_{n-1})+\prod_{i=1}^{n-1}(x_i+\bar{v}_i).
\end{equation}
Similarly, if the companion states ${\bf v}$ and $\tilde{{\bf v}}$ are in two distinct cycles, then interchanging their predecessors joins the two cycles.
If either process continues until all of the cycles in $\Omega(f)$ can be joined into a single cycle, then we obtain a de Bruijn sequence. This construction is the {\it cycle joining method} (CJM).

Given an FSR with feedback function $f$, its {\it adjacency graph} $G_f$, or simply $G$ if $f$ is clear, is an undirected multigraph whose vertices correspond to the cycles in $\Omega(f)$. Two distinct vertices are {\it adjacent} if they share a conjugate (or companion) pair. The number of edges between them is the number of shared conjugate (or companion) pairs, with a specific pair assigned to each edge. We know from \cite{AEB87} that there is a bijection between the set of spanning trees of $G_f$ and the set of all inequivalent de Bruijn sequences constructible by the CJM on input $f$.

We now introduce two simple FSRs that we will often use. The {\it pure cycling register} (PCR) {\it of order} $n$ is an LFSR with feedback function and characteristic polynomial
\begin{equation}\label{eq:PCR}
f_{\rm PCR}(x_0,x_1,\ldots,x_{n-1})= x_0 \mbox{ and }
f_{\rm PCR}(x)=x^n+1.
\end{equation}
Each cycle generated by the PCR of order $n$ is $n$-periodic and has the form $(c_0,c_1,\ldots,c_{n-1})$. The {\it complemented cycling register} (CCR) {\it of order} $n$ is an NLFSR with feedback function
\begin{equation}\label{eq:CCR}
f_{\rm CCR}(x_0,x_1,\ldots,x_{n-1})= x_0 + 1.
\end{equation}
Each cycle generated by the CCR is $2n$-periodic and has the form
\begin{equation}\label{ccr-seq}
(c_0,c_1,\ldots,c_{n-1},\bar{c}_0,\bar{c}_1,\ldots,\bar{c}_{n-1}),
\end{equation}
that is, we always have $c_i=\bar{c}_{i+n}$ for all $i \geq 0$. The {\it weight} of a cycle $C$ or a state $\mathbf{v}$, denoted respectively by ${\rm wt}(C)$ and ${\rm wt}(\mathbf{v})$, is the number of $1$s in the cycle or the state. It is clear that the weight of each CCR cycle is $n$.

For a given order $n$, the usual cycle representatives of a cycle in $\Omega(f_{\rm PCR})$ and a cycle in $\Omega(f_{\rm CCR})$ are called the {\it necklace} and the {\it co-necklace}, respectively. It is known, for example in \cite[Algorithm~2]{Gabric18}, that testing if a state is the necklace or the co-necklace of a cycle takes $O(n)$ time and $O(n)$ space. Most fast algorithms to generate de Bruijn sequences work on either the PCR or the CCR.

Although the CCR is not linear, the corresponding sequences can be generated by an LFSR. Each sequence $\mathbf{s} = (s_0,s_1,\ldots,s_{2n-3})$ of the CCR of order $n-1$, with $n \geq 3$, satisfies
\[
s_{n-1+i}=s_i+1 \  \mbox{and} \ s_{n+i}=s_{i+1}+1, \ \mbox{for} \ i\geq 0.
\]
Combining them, we obtain the relation
\[
s_{n+i}=s_i+s_{i+1}+s_{n-1+i}, \ \mbox{for} \ i\geq 0.
\]
Hence, $\mathbf{s}$ can be generated by an LFSR of order $n$ with characteristic polynomial and feedback function
\begin{equation}\label{equ:1}
h(x) = x^{n} + x^{n-1} + x + 1 = (x^{n-1}+1) (x+1) \mbox{ and }
h(x_0,x_1,\ldots,x_{n-1})=x_0+x_1+x_{n-1}.
\end{equation}

The LFSR with characteristic polynomial $h(x)$ has several good properties.
The authors of \cite{Salasame} used it to construct a successor rule that generates the {\tt prefer-same} de Bruijn sequence in $O(n)$ time and $O(n)$ space per bit. They named the said LFSR the {\it pure run-length register} (PRR). We henceforth adopt the name and refer to $h(x)$ in Equation (\ref{equ:1}) as $f_{\rm PRR, \it n}(x)$. When $n$ is clear in the context, we also use the abbreviation $f_{\rm PRR}(x)$. The same practice of specifying the order for precision when necessary applies to the other LFSRs.

\begin{lemma}\label{lem:1}
Each cycle in $\Omega(f_{\rm PRR})$ of order $n$ is either a PCR cycle or a CCR cycle of order $n-1$.
\end{lemma}

\begin{proof}
We have seen that all sequences of the CCR of order $n-1$ can be generated by the PRR of order $n$. Hence, $G(f_{{\rm CCR},n-1})\subseteq G(f_{{\rm PRR},n})$. For two LFSRs with respective characteristic polynomials $f_1(x)$ and $f_2(x)$, one has $G(f_1)\subseteq G(f_2)$ if and only if $f_1(x)$ divides $f_2(x)$~\cite[Lemma 4.2 (a)]{GG05}. Since $(x^{n-1}+1)$ divides $f_{{\rm PRR},n}(x)$, we have $G(f_{{\rm PCR},n-1}) \subseteq G(f_{{\rm PRR},n})$. Hence, 
\[
G(f_{{\rm PCR},n-1})\cup G(f_{{\rm CCR},n-1})\subseteq G(f_{{\rm PRR},n}).
\]
Since $G(f_{{\rm PCR},n-1})$ and $G(f_{{\rm CCR},n-1})$ are disjoint, with
$|G(f_{{\rm PCR},n-1})|=|G(f_{{\rm CCR},n-1})|=2^{n-1}$, and
$|G(f_{{\rm PRR},n})|=2^n$, it is clear that
\[
G(f_{{\rm PRR},n}) = G(f_{{\rm PCR},n-1}) \cup G(f_{{\rm CCR},n-1}).
\]
Thus, each cycle in $\Omega(f_{{\rm PRR},n})$ is a cycle either in
$\Omega(f_{{\rm PCR},n-1})$ or in $\Omega(f_{{\rm CCR},n-1})$.
\end{proof}

\begin{remark}
Let an $n$-stage state $c_0,c_1,\ldots,c_{n-1}$ of a cycle $C$ in $\Omega(f_{\rm PRR})$ be given. If $c_0=c_{n-1}$, then, by the proof of Lemma~\ref{lem:1}, $C$ is a PCR cycle. Otherwise, $C$ is a CCR cycle. This is why we sometimes refer to the PRR in \cite{Salasame} as a {\rm pure and complemented cycling register} (PCCR).
\end{remark}

If $\phi(\cdot)$ is the Euler totient function, then the number of cycles in $\Omega(f_{{\rm PRR},n})$ is
\begin{align}\label{eq:number}
\bar{Z}_n &= Z_{n - 1} + Z_{n - 1}^{*} \mbox{, where} \notag \\
Z_{n-1} &= {\frac{1}{n-1}} \left({\sum_{d|(n-1)}}\phi(d) \,2^{\frac{n-1}{d}}\right)
\mbox{ and }
Z_{n-1}^{*} = {\frac {1}{2(n-1)}}
\left(\sum_{\substack{d \mid (n-1)\\ d \text{ odd}}} \phi(d) \, 2^{\frac{n-1}{d}}\right)
\end{align}
are the respective number of cycles in $\Omega(f_{{\rm PCR},n-1})$ and in $\Omega(f_{{\rm CCR},n-1})$ \cite{Fred82}. A proof for the formula of $Z_n$ and a sketch of the proof for the formula of $Z_{n}^{*}$ were both due to Golomb~\cite{Golomb}. A more thorough discussion was supplied by Sloane in~\cite[Section~3]{Sloane02}.

We partition the cycles in $\Omega(f_{{\rm PRR},n})$ into two parts, namely the PCR cycles $\CP_1,\ldots,\CP_{Z_{n-1}}$ and the CCR cycles $\CC_1,\ldots,\CC_{Z_{n-1}^*}$. Excluding the cycle $(1^{n-1})$ whose representative is obviously $1^n$, the other cycle representatives must be $n$-stage states whose forms are either
\[
0,c_1,\ldots,c_{n-2},0\ \mbox{ or } \ 0,c_1,\ldots,c_{n-3},0,1,
\]
where $0,c_1,\ldots,c_{n-2}$ is either the necklace in a PCR cycle or the co-necklace in a CCR cycle.

The PRR has yet another interesting property. The {\it run-length encoding} of an $n$-stage state $\mathbf{v}=v_0,v_1,\ldots,v_{n-1}$ is a compressed representation that stores, consecutively, the lengths of the maximal runs of each element. The {\it run-length} of $\mathbf{v}$ is the length of its run-length encoding. For example, the state $0001011101$ has run-length encoding $311311$ and run-length 6. The following lemma was established in~\cite{Salasame}.

\begin{lemma}\label{lem:2}
All states in a cycle in $\Omega(f_{\rm PRR, n})$ have the same run-length for a given $n$.
\end{lemma}

\begin{example}\label{ex:PCCR}
Let $n=6$. The $12$ cycles in $\Omega(f_{\rm PRR})$ consists of the $8$ cycles generated by the PCR of order $5$, namely
\begin{align*}
&\CP_1:=(00000), ~\CP_2:=(00001), ~ \CP_3:=(00011), ~ \CP_4:=(00101), \\
&\CP_5:=(00111), ~ \CP_6:=(01011), ~ \CP_7:=(01111), ~ \CP_8:=(11111),
\end{align*}
and the $4$ cycles generated by the CCR of order $5$, namely
\[
\CC_1:=(00000 11111), ~\CC_2 :=(0001011101), ~ \CC_3:=(0010011011), ~ \CC_4:=(0101010101).
\]
The cycles are presented in increasing lexicographical order within their respective types as
\[
\CP_1 \prec_{\rm lex} \CP_2 \prec_{\rm lex} \ldots \prec_{\rm lex} \CP_8
\mbox{ and }
\CC_1 \prec_{\rm lex} \ldots \prec_{\rm lex} \CC_4.
\]
The cycle representatives of $\CP_1,\ldots,\CP_8$ and $\CC_1,\ldots,\CC_4$ are, in that order,
\begin{align*}
000000,~000010, ~000110,~001010,~001110,~010110,~011110,~111111,\\
000001, ~000101,~001001,~010101.
\end{align*}
In $\CC_2$, there are $10$ distinct $6$-stage states:
\[
000101, 001011, 010111, 101110, 011101, 111010, 110100, 101000, 010001, 100010,
\]
and their run-lengths are, respectively,
\[
3111, 2112, 1113, 1131, 1311, 3111,2112, 1113, 1131, 1311.
\]
All of them have the same run-length $4$.
\end{example}

Gabric et al. in~\cite{Gabric18} and Sawada et al. in~\cite{Sawada16} proposed several fast algorithms to generate de Bruijn sequences. In those papers they ordered the cycles in $\Omega(f_{\rm PCR})$ and in $\Omega(f_{\rm CCR})$ lexicographically according to how each cycle's necklace or co-necklace compares to one another, respectively. In each case, they replace the usual FSR-based generating algorithm by a well-chosen mechanism, called the {\it successor rule} $\rho(x_0,x_1,\ldots,x_{n-1})$, to generate the next bit. Given an FSR with a feedback function $f(x_0,x_1,\ldots,x_{n-1})$, the general thinking behind this approach is to determine some condition which guarantees that the resulting sequence is de Bruijn. 

Let $A$ be a nonempty subset of the set of all states. For each state $\mathbf{c}=c_0,c_1,\ldots,c_{n-1}$, a nontrivial successor rule $\rho$ is an assignment 
\begin{equation}\label{equ:rho1}
\rho(c_0,c_1,\ldots,c_{n-1}) :=
\begin{cases}
\overline{f(c_0,\ldots,c_{n-1})} & \mbox{if } \mathbf{c} \in A,\\
f(c_0,\ldots,c_{n-1}) & \mbox{if } \mathbf{c} \notin A.
\end{cases}
\end{equation}
To be precise, the successor of $\mathbf{c}=c_0,c_1,\ldots,c_{n-1}$ is $c_1,\ldots,c_{n-1}, f(c_0,\ldots,c_{n-1})$ except when $\mathbf{c} \in A$. When $\mathbf{c} \in A$, the successor is $c_1,\ldots,c_{n-1}, f(c_0,\ldots,c_{n-1})+1$, that is, the last bit of the successor is the complement of the last bit of the successor when $\mathbf{c} \notin A$. Our interest is to characterize a set $A$ whose corresponding successor rule ensures that the cycles can be joined to form a de Bruijn sequence. A framework to judge whether a successor rule can generate a de Bruijn sequence was provided in \cite[Section 3]{Gabric18}. We restate their result in our notation as follows.

\begin{theorem}[Theorem 3.5 in \cite{Gabric18}] \label{thm:old}
Let $\prec$ denote a valid order on the cycles of an FSR whose feedback function is $f(x_0,x_1,\ldots,x_{n-1})$. Let the cycles be ordered as $C_1 \prec C_2 \prec \ldots \prec C_r$. Let $A$ be a set that contains all states which constitute $r-1$ conjugate pairs with the following properties. For each cycle $C_i$ with $1 < i \leq r$, there exists a unique conjugate pair $(\mathbf{c}, \hat{\mathbf{c}})$ such that $\mathbf{c}$ is in $C_i$, $\hat{\mathbf{c}}$ is in $C_j$, where $j < i$, and both $\mathbf{c},\hat{\mathbf{c}} \in A$. The successor rule $\rho$ in Equation (\ref{equ:rho1}) generates a de Bruijn sequence.
\end{theorem}

The theorem states that, if one can define a suitable order on the cycles generated by a given FSR, then a successor rule can be devised to generate a de Bruijn sequence by constructing a spanning tree.

If in each cycle $C_i$ with $1 < i \leq r$ one can uniquely determine a state whose conjugate state is in another cycle $C_j \prec C_i$, then the successor rule interchanges their respective successor states. The successor of any other state which is not identified as either the unique state in each $C_i$ or its conjugate state in $C_j$ remains to be the one assigned by the FSR with feedback function $f$. This process joins cycles $C_i$ and $C_j$. All of the cycles are eventually joined into a single cycle as $i$ goes from $2$ to $r$.

The set $A$ in Theorem \ref{thm:old} determines $r-1$ conjugate pairs that correspond to $r-1$ edges in the adjacency graph $G$ that has $r$ vertices. 
If an edge connects two distinct cycles $C_i$ and $C_j$ with $r \geq i > j \geq 1$, then the direction of this edge is from $C_i$ to $C_j$. The condition guarantees that there is a unique path from each $C_i$ to $C_1$ for $1 < i \leq r$ and a rooted spanning tree is explicitly constructed. This ensures that the generated sequence is a de Bruijn sequence. Henceforth, we call such a set $A$ \textit{a critical set of spanning conjugate pairs} or \textit{a critical set}, in short.

The following result from \cite{Jansen91} is useful to tell when a given cycle is lexicographically less than another cycle.

\begin{lemma}\label{lem:Jansen} \cite{Jansen91}
If $C$ is a cycle of an FSR whose cycle representative is a nonzero state $\mathbf{v}$, then the companion state $\tilde{\mathbf{v}}$ is in another cycle $C'$ with
$C' \prec_{\rm lex} C$.
\end{lemma}

Sala used the PRR of order $n$ to construct a successor rule that generates a de Bruijn sequence in $O(n)$ time and $O(n)$ space per bit in \cite{SalaTh}. The original proof was quite involved. Here we recall the result and supply a simpler proof.

\begin{theorem}\label{thm:sala}\cite{SalaTh}
For each state $\mathbf{c}=c_0,c_1,\ldots,c_{n-1}$, let 
${\bf u}_\mathbf{c}:=c_1,c_2,\ldots,c_{n-1}$. The successor rule
\begin{equation}\label{equ:sala}
\varsigma(\mathbf{c})=
\begin{cases}
\overline{c_0+c_1+c_{n-1}} & \mbox{if } {\bf u}_{\mathbf{c}} \mbox{ is a necklace or a co-necklace},\\
c_0+c_1+c_{n-1} & \mbox{otherwise,}
\end{cases}
\end{equation}
generates a de Bruijn sequence of order $n$.
\end{theorem}

\begin{proof}
Let the cycles in $\Omega(f_{{\rm PRR},n})$ be ordered as $(0^n)=C_0 \prec_{\rm lex} C_1 \prec_{\rm lex} \ldots \prec_{\rm lex} C_{\bar{Z}_n-1}$. For a given state 
$\mathbf{c}$ in some $C_i$ with $i > 0$, let $c_{n} := c_0+c_1+c_{n-1}$ and $\mathbf{v}:=c_1,c_2,\ldots,c_{n}$. As we have discussed earlier, the $(n-1)$-stage state ${\bf u}_{\mathbf{c}}$ is a necklace or a co-necklace if and only if the $n$-stage state $\mathbf{v}$ is the cycle representative of $C_i$, which is unique. By Lemma~\ref{lem:Jansen}, the companion state $\tilde{\mathbf{v}}$ must be in $C_j \neq C_i$ with $C_j \prec_{\rm lex} C_i$, that is, $j < i$. Collecting the two states in the conjugate pair $({\bf c},\hat{\bf c})$ as $i$ ranges from $1$ to $\bar{Z}_n-1$ yields a critical set $A$. The desired conclusion follows. 
\end{proof}

Now that prior arts have been covered, we are ready to present our new successor rules. Each of the next two sections introduces a new generic class of successor rules. We use Theorem~\ref{thm:old} to confirm that these two classes generate exponentially many de Bruijn sequences.

\section{The first class of successor rules from the PRR}\label{sec:class1}

We begin by giving a general formula for successor rules to generate de Bruijn sequences based on the PRR before defining a new class of successor rules explicitly.

\begin{theorem}\label{thm:rho-1}
For each state ${\bf c}=c_0,c_1,\ldots,c_{n-1}$ produced by the $PRR$ of order $n \geq 3$, we define the state $\mathbf{v}_{\mathbf{c}} := c_1,\ldots,c_{n-1},1$. The states forming the conjugate pair $({\bf c}, \hat{{\bf c}})$ belong to a critical set $A_{\Psi}$ if one of the followings holds.
\begin{enumerate}
	\item[(i)] The state $\mathbf{v}_{\mathbf{c}}$ is the cycle representative of a CCR cycle $\CC$ if $c_1=0$.
	\item[(ii)] The state $\mathbf{v}_{\mathbf{c}}$ can be uniquely determined in a PCR cycle $\CP$ if $c_1=1$.
\end{enumerate}
The successor rule
\begin{equation}\label{sr-1}
\Psi({\bf c}) :=
\begin{cases}
\overline{c_{0} + c_{1} + c_{n - 1} }  &  \mbox{if }  {\bf c} \in A_{\Psi},\\
c_{0} + c_{1} + c_{n - 1}     &  \mbox{otherwise,}
\end{cases}
\end{equation}
generates a de Bruijn sequence of order $n$.
\end{theorem}

\begin{proof}
Suppose that the cycles in $\Omega(f_{\rm PRR})$ have been ordered lexicographically as
\[
(0^{n-1}) = C_0 \prec_{\rm lex} C_1 \prec_{\rm lex} \ldots \prec_{\rm lex} C_{\bar{Z}_n-1}.
\]
The state $\mathbf{v}_{\mathbf{c}}$ that satisfies either one of the above conditions, depending on the value of $c_1$, determines a conjugate pair
$({\bf c}, \hat{{\bf c}})$. For $ 1 \leq i < \bar{Z}_n$, Theorem \ref{thm:old} requires us to show that each $C_i$ has a uniquely identified state whose conjugate state is in $C_j$, with $C_j \prec_{\rm lex} C_i$. This is equivalent to showing that each $C_i$ has a uniquely identified state whose companion state is in $C_j$, with $C_j \prec_{\rm lex} C_i$. We note that each $C_i$ must contain at least one state whose last bit is $1$, that is, $C_i$ always contains a state that can serve as $\mathbf{v}_{\mathbf{c}}$.
	
If $\mathbf{v}_{\mathbf{c}}$ is the cycle representative of a CCR cycle $\CC$, then it is uniquely determined and its companion state $\tilde{\mathbf{v}}_{\mathbf{c}}$ must be in a PCR cycle $\CP$. By Lemma~\ref{lem:Jansen}, we have $\CP \prec_{\rm lex} \CC$.

Let $\mathbf{v}_{\mathbf{c}}$ be a uniquely determined state in a PCR cycle $\CP$. If $\CP = (1^{n-1})$, then $\mathbf{v}_{\mathbf{c}} = 1^{n}$ and $\tilde{\mathbf{v}}_{\mathbf{c}} = 1^{n-1}0$ is in the CCR cycle $(0^{n-1}1^{n-1}) \prec_{\rm lex} (1^{n-1}) = \CP$. If $\CP \neq (1^{n-1})$, then the unique cycle representative of $\CP$ begins with a $0$ and has the form
\[
c_j,\ldots,c_{n-1},1, c_2\ldots,c_j
\]
for some $j$ in the range $2 \leq j < n$. The state $\tilde{\mathbf{v}}_{\mathbf{c}}$ is in a CCR cycle $\CC$ that contains the state
\[
c_j,\ldots,c_{n-1},0,\bar{c}_2,\ldots,\bar{c}_j,
\]
which is clearly lexicographically less than the cycle representative of $\CP$. Hence,
$\CC \prec_{\rm lex} \CP$. Thus, $\Psi$ generates a de Bruijn sequence of order $n$ by Theorem~\ref{thm:old}.
\end{proof}

We emphasize that the uniquely identified state in each CCR cycle must be the cycle representative. The uniquely determined state in each nonzero PCR cycle can be any state $\mathbf{v}_{\mathbf{c}}$ that ends with a $1$ as long as there is a way to uniquely identify it. Different ways to uniquely identify $c_1=1,c_2,\ldots,c_{n-1},1$ in each PCR cycle yield distinct successor rules. Each such rule generates a de Bruijn sequence. The rest of this section supplies two clusters of successor rules based on concrete choices for the critical set $A_{\Psi}$.

Since $c_1=1$, one strategy is to uniquely identify the $(n-1)$-stage state $c_2,\ldots,c_{n-1},1$ in any nonzero cycle produced by the PCR of order $n-1$ with respect to its necklace. Every state in a PCR cycle can be transformed into the necklace by repeated left shift operations. Now we define a new operator $\Lambda$ which consists of a sequence of left shift operations so that the first $1$ is cyclically left-shifted to the end. Formally, given a nonzero state ${\bf u}_{\mathbf{c}} =c_1, c_2,\ldots,c_{n-1}$, we choose $i$ with $1 \leq i <n$ to be the smallest index for which $c_i = 1$. We define
\[
\Lambda \,\mathbf{u}_{\mathbf{c}} :=c_{i+1},\ldots,c_{n-1},c_1,\ldots,c_{i}
\]
and use the notation $\Lambda^r\, \mathbf{u}_{\mathbf{c}}$ to mean $\Lambda^{r-1} (\Lambda \, \mathbf{u}_{\mathbf{c}})$, with $\Lambda^0 \, \mathbf{u}_{\mathbf{c}} := \mathbf{u}_{\mathbf{c}}$. In a PCR cycle, $\Lambda$ transforms $\mathbf{u}_{\mathbf{c}}$ to the necklace after at most $j$ steps, where $j$ is the weight of the cycle. In $(01011)$, for example, if $\mathbf{u}_{\mathbf{c}} = 01101$, then $\Lambda \, \mathbf{u}_{\mathbf{c}} = 10101$ and $\Lambda^2 \, \mathbf{u}_{\mathbf{c}} = 01011$ is already the necklace.

All ingredients to explicitly construct successor rules in the class $\Psi$ are now in place and distinct ways to determine the desired state in any PRR cycle can be explicitly written.

\begin{theorem}\label{prop:1}
Let positive integers $n$, $t$, and $k_1, \ldots, k_t$ be such that
\[
n \geq 3, \quad 2 \leq t \leq n-1, \quad 1 = k_1 < k_2 < \ldots < k_{t} =n, \quad \mbox{and } k_{t-1} < n-1.
\]
For each state $\mathbf{c}=c_0,c_1,\ldots,c_{n-1}$, let
$\mathbf{u}_{\mathbf{c}} := c_1 ,c_2,\ldots,c_{n-1}$. The successor rule
\begin{equation}\label{new3}
\Psi_1 (\mathbf{c}) =
\begin{cases}
\overline{c_{0} + c_{1} + c_{n - 1}} & \mbox{if } c_1=0 \mbox{ and } \mathbf{u}_{{\mathbf{c}}}  
\mbox{ is a co-necklace}, \\
\overline{c_{0} + c_{1} + c_{n - 1} } & \mbox{if } c_1=1 \mbox{ and there is an } i \mbox{ such that } \\
& \quad k_i \leq {\rm wt}(\mathbf{u}_{{\mathbf{c}}}) < k_{i+1} \mbox{ and } \Lambda^{k_i} \, \mathbf{u}_{\mathbf{c}} \mbox{ is a necklace,} \\
c_{0} + c_{1} + c_{n - 1}     & \mbox{otherwise},
\end{cases}
\end{equation}
generates de Bruijn sequences of order $n$.
\end{theorem}

\begin{proof}
We prove correctness by showing that the conditions required by Theorem \ref{thm:rho-1} are satisfied. Let ${\bf v}_{\mathbf{c}}:={\bf u}_{\mathbf{c}},1=c_1,c_2,\ldots,c_{n-1},1$ and ${\bf x}_{\mathbf{c}}:=0,{\bf u}_{\mathbf{c}} = 0,c_1,c_2,\ldots,c_{n-1}$. 
If $c_1=0$, then ${\bf v}_{\mathbf{c}}$ is the cycle representative of a CCR cycle $\CC$ of order $n-1$ if and only if ${\bf u}_{\mathbf{c}} = 0,c_2,\ldots,c_{n-1}$ is the co-necklace of $\CC$. If $c_1=1$, then ${\bf v}_{\mathbf{c}}$ is a uniquely determined state in a PCR cycle $\CP$ of order $n-1$ if and only if $\Lambda \, {\bf u}_{\mathbf{c}} = 
L {\bf u}_{\mathbf{c}} = c_2,\ldots,c_{n-1},1$ is uniquely determined in $\CP$. Since $k_i \leq {\rm wt}({\bf u}_{\mathbf{c}}) < k_{i+1}$ and $\Lambda^{k_i}\, {\bf u}_{\mathbf{c}}$ is the necklace of $\CP$, we confirm that ${\bf u}_{\mathbf{c}}$ is uniquely determined.
\end{proof}

Because the weight of the PCR cycle $(1^{n-1})$ is $n-1$, we always take $k_t=n$ for consistency in formulating the successor rule. Letting $k_{t-1}=n-1$ does not add any value since $(1^{n-1})$ has just one state. The number of inequivalent de Bruijn sequences that Theorem \ref{prop:1} generates is established as the next result.

\begin{proposition}\label{lem:3}
Taking all valid parameters $t$ and $\{k_1,k_2,\ldots,k_t\}$ in the statement of Theorem~\ref{prop:1}, the number of de Bruijn sequences generated by $\Psi_1$ in Equation
(\ref{new3}) is $2^{n-3}$.
\end{proposition}

\begin{proof}
For each choice of $t$ and $\{k_1,k_2,\ldots,k_t\}$, with $k_1=1$, $k_{t-1}<n-1$, and $k_t=n$, we obtain a critical set of conjugate pairs to join all of the cycles. Distinct choices of $t$ and $\{k_1,k_2,\ldots,k_t\}$ yield inequivalent de Bruijn sequences. The total number of choices is
\begin{equation}\label{eq:numPsi1}
\binom{n-3}{0} + \binom{n-3}{1}  + \ldots + \binom{n-3}{n-3} = \sum_{j=0}^{n-3} 
\binom{n-3}{j} = 2^{n-3}.
\end{equation}
\end{proof}

The next theorem presents another way to uniquely identify a state in each cycle and obtain the corresponding successor rule.

\begin{theorem}\label{prop:2}
Let  $\Delta := \lcm(1,2,\ldots,n-2)$ and $1 \leq k \leq \Delta$. For each state $\mathbf{c} = c_0,c_1,\ldots,c_{n-1}$, let $\mathbf{u}_{\mathbf{c}} :=c_1,c_2,\ldots,c_{n-1}$.
The successor rule
\begin{equation}\label{new4}
\Psi_2 (\mathbf{c}) =
\begin{cases}
\overline{c_{0} + c_{1} + c_{n - 1} } & \mbox{if } c_1=0\ {\rm and }\ \mathbf{u}_{\mathbf{c}}  \mbox{ is a co-necklace}, \\
\overline{c_{0} + c_{1} + c_{n - 1} } & \mbox{if } c_1=1\ {\rm and }\ \Lambda^{k} \, \mathbf{u}_{\mathbf{c}} \mbox{ is a necklace}, \\
c_{0} + c_{1} + c_{n - 1}    & \mbox{otherwise},
\end{cases}
\end{equation}
generates a total of
\begin{equation}\label{eq:numPsi2}
\lcm(1,2,\ldots,n-2)  \geq 2^{n-3}
\end{equation}
de Bruijn sequences of order $n$ by varying $k$.
\end{theorem}

\begin{proof}
The proof of the first part is similar to the proof of Theorem \ref{prop:1} and is, therefore, omitted here. To enumerate these de Bruijn sequences, we note that the weights of the PCR cycles vary from $1$ to $n-2$, except for $(1^{n-1})$ which has only one state. The weight of each $\CP \neq (1^{n-1})$ gives the number of states that matters. For any given $k > 0$, let $k'=k-1$ and $\mathbf{w}_{\mathbf{c}}= \Lambda \, \mathbf{u}_{\mathbf{c}}= L \mathbf{u}_{\mathbf{c}}$. We consider the system of congruences
\begin{equation}\label{cong}
\{ k' \equiv a_i \Mod{i} \mid 1 \leq i \leq (n-2)\}.
\end{equation}
Applying $\Psi_2$, if $\Lambda^k \, \mathbf{u}_{\mathbf{c}} = \Lambda^{k'} \, \mathbf{w}_{\mathbf{c}}$ is a necklace, then $\Lambda^{a_i} \,\mathbf{w}_{\mathbf{c}}$ is the necklace in a cycle of weight $i$. Hence, $a_i$ uniquely identifies a state in the cycle. So each possible choice of $\{a_1,\ldots,a_{n-2}\}$ uniquely identifies a collection of states, each belonging to a cycle. By the Chinese Remainder Theorem \cite{Ding}, the number of distinct choices of $\{a_1,\ldots,a_{n-2}\}$ is $\lcm(1,2,\ldots,n-2)$, which is obtained as $k$ runs from $1$ to $\Delta$. Hence, there are $\lcm(1,2,\ldots,n-2)$ de Bruijn sequences.  By  \cite[Section~2]{Farhi}, we obtain
\begin{equation}\label{eq:NumOfPsi2}
\lcm(1,2,\ldots,n-2) \geq  (n-2)  \binom{n-3}{\left\lfloor \frac{n-3}{2} \right\rfloor} \geq 2^{n-3}.
\end{equation}
\end{proof}

\begin{example}\label{ex:rho1_n6}
We continue from Example~\ref{ex:PCCR} to consider $\Psi$ for $n=6$. Theorem~\ref{prop:1} provides $8$ distinct successor rules. The resulting $8$ distinct de Bruijn sequences are listed in the first part of Table \ref{table:1}. Applying Theorem~\ref{prop:2}, again for $n=6$, yields the $12$ distinct de Bruijn sequences listed in the second part of Table~\ref{table:1}. For ease of comparison, the initial state is fixed to be $000000$.

We note that Theorems~\ref{prop:1} and \ref{prop:2} may produce equivalent sequences. In Table \ref{table:1}, the sequences in entries $1$, $2$, $3$, and $8$ based on Theorem \ref{prop:1} are the same as the sequences in entries $2$, $3$, $4$, and $1$ based on Theorem \ref{prop:2}, respectively. Table~\ref{table:1} contains $16$ distinct de Bruijn sequences in total.
	
\begin{table*}[h!]
	\caption{Inequivalent de Bruijn sequences constructed based on Theorems~\ref{prop:1} and~\ref{prop:2} with $n=6$.}
	\label{table:1}
	\renewcommand{\arraystretch}{1.2}
	\centering
	{\footnotesize
		\begin{tabular}{ccc}
			\hline
			Entry & $\{k_1,k_2,\ldots,k_t\}$ & The resulting sequence based on Theorem~\ref{prop:1}\\
			\hline
			$1$   & $\{1,6\}$                   & $(0000001111110000101111011101000110001001110011011001010110101001)$\\
			$2$   & $\{1,2,6\}$                 & $(0000001111110001100001011101011010010101000100110111101100111001)$\\
			$3$   & $\{1,3,6\}$                 & $(0000001111110011100001011101111010001100010011011010110010101001)$\\
			$4$   & $\{1,4,6\}$                 & $(0000001111110111100001011101000110001001110011011001010110101001)$\\
			$5$   & $\{1,2,3,6\}$               & $(0000001111110011100011000010111011110100101010001001101101011001)$\\
			$6$   & $\{1,2,4,6\}$               & $(0000001111110111100011000010111010110100101010001001101100111001)$\\
			$7$   & $\{1,3,4,6\}$               & $(0000001111110111100111000010111010001100010011011010110010101001)$\\
			$8$   & $\{1,2,3,4,6\}$             & $(0000001111110111100111000110000101110100101010001001101101011001)$\\
			\hline
			\hline
			Entry & $k$ & The resulting sequence based on Theorem~\ref{prop:2}\\
			\hline
			$1$   & $1$ & $(0000001111110111100111000110000101110100101010001001101101011001)$\\
			$2$   & $2$  & $(0000001111110000101111011101000110001001110011011001010110101001)$\\
			$3$   & $3$  & $(0000001111110001100001011101011010010101000100110111101100111001)$\\
			$4$   & $4$  & $(0000001111110011100001011101111010001100010011011010110010101001)$\\
			$5$   & $5$  & $(0000001111110111100011000010111010010101101010001001110011011001)$\\
			$6$   & $6$  & $(0000001111110000101111011101011010001100010011011001110010101001)$\\
			$7$   & $7$  & $(0000001111110011100011000010111010010101000100110111101101011001)$\\
			$8$   & $8$  & $(0000001111110000101110111101000110001001110011011001010110101001)$\\
			$9$   & $9$  & $(0000001111110111100011000010111010110100101010001001101100111001)$\\
			$10$   & $10$  & $(0000001111110011100001011110111010001100010011011010110010101001)$\\
			$11$   & $11$  & $(0000001111110001100001011101001010110101000100111001101111011001)$\\
			$12$   & $12$ & $(0000001111110000101110111101011010001100010011011001110010101001)$\\
			\hline
			
		\end{tabular}
	}
\end{table*}
\end{example}

There might be other ways to find more critical sets by uniquely identifying a state whose last bit is $1$ in each of the PCR cycles, leading to more successor rules. Interested readers are invited to invent their favourites.

\section{The second class of successor rules from the PRR}\label{sec:class2}

In this section we discuss another class of successor rules.

\begin{theorem}\label{thm:rho-2}
Let $n\geq 3$ be given. For each state ${\bf c}=c_0,c_1,\ldots,c_{n-1}$ produced by
$f_{\rm PRR}$, we define $\mathbf{x}_{\mathbf{c}}$ based on the value of $c_{n-1}$ as
\[
{\bf x}_{\mathbf{c}}:=
\begin{cases}
	0,c_1,c_2,\ldots,c_{n-1} & \mbox{if } c_{n-1}=0,\\
	\bar{c}_1,\bar{c}_2,\ldots,\bar{c}_{n-2}, 0,c_1 & \mbox{if } c_{n-1}=1.
\end{cases}
\]	
The states forming the conjugate pair $({\bf c}, \hat{{\bf c}})$ belong to a critical set $A_{\Upsilon}$ if exactly one of the following conditions is satisfied. 
\begin{enumerate}
\item[(i)] The state ${\bf x}_{\mathbf{c}}$ is a uniquely determined state whose first bit is $0$ in a PCR cycle $\CP$ if $c_{n-1}=0$.
\item[(ii)]  The state ${\bf x}_{\mathbf{c}}$ is the cycle representative of a CCR cycle $\CC$ if $c_{n-1}=1$.
\end{enumerate}
The successor rule
\begin{equation}\label{sr-2}
\Upsilon(\mathbf{c}) :=
\begin{cases}
\overline{c_{0} + c_{1} + c_{n - 1}}  & \mbox{if } {\bf c} \in A_{\Upsilon},\\
c_{0} + c_{1} + c_{n - 1}     & \mbox{otherwise,}
\end{cases}
\end{equation}
generates a de Bruijn sequence of order $n$.
\end{theorem}

\begin{proof}
We prove that we can define an order $\prec$ on all cycles of the PRR and  $A_{\Upsilon}$  is a critical set. For any two distinct CCR cycles $\CC_i$ and $\CC_j$,  we define $\CC_i  \ { \prec} \  \CC_j$  if $\CC_i  \ { \prec}_{\rm lex} \CC_j$. This allows us to order all of the CCR cycles lexicographically as 
\[
(0^{n-1}1^{n-1}) = \CC_1 \prec \CC_2 \prec \cdots \prec \CC_{t} \mbox{ with } t:=Z_{n-1}^*.
\]

Our next task is to include the other cycles. It is clear that each cycle, except for $(1^{n-1})$, has a unique state ${\bf x}_{\mathbf{c}}$ satisfying exactly one of the two conditions $(i)$ and $(ii)$.
 
If ${\bf x}_{\mathbf{c}}$ satisfies condition $(ii)$, then ${\bf x}_{\mathbf{c}}$ belongs to a CCR cycle $\CC$. Hence, either ${\bf c}$ or $\hat{{\bf c}}$ belongs to the same CCR cycle $\CC$. Without loss of generality, we assume ${\bf c} \in \CC$. If $\CC  = (0^{n-1}1^{n-1})$, which is the lexicographically least among all of the CCR cycles, then the cycle representative is ${\bf x}_{\mathbf{c}}=0^{n-1}1$. In this case, we have ${\bf c}=01^{n-1}$ and, thus, $\hat{{\bf c}}$ is in the PCR cycle $(1^{n-1})$. We order $(1^{n-1}) \prec (0^{n-1}1^{n-1})$. If ${\bf c} \in \CC  \neq (0^{n-1}1^{n-1})$, then $c_{n-1} =1$ and $c_0 =0$. Since ${\bf x}_{\mathbf{c}} = \bar{c}_1,\bar{c}_2,\ldots,\bar{c}_{n-2},0,c_1 $ is the cycle representative of a CCR cycle $\CC$, we must have $\bar{c}_1 =0$ and, hence, $c_1 = c_{n-1} = 1$. Thus, $\hat{{\bf c}} $ must be in some PCR cycle $\CP=(1,c_1,\ldots,c_{n-1})$ and we say that $\CC$ is adjacent to $\CP$. 

On the other hand, if ${\bf x}_{\mathbf{c}} = {\bf c}$ satisfies condition $(i)$, then   ${\bf c}$ belongs to a PCR cycle $\CP$. Hence, $\hat{{\bf c}}$ is in a CCR cycle $\CC'$. When this is the case, we say that $\CP$ is adjacent to $\CC'$. 
 
If a CCR cycle $\CC \neq(0^{n-1}1^{n-1})$ is adjacent to a PCR cycle $\CP$ and $\CP$ is adjacent to another CCR cycle $\CC'$, then $\CC'{ \prec}_{\rm lex} \CC$. From the above discussion, if ${\bf c}=c_0,c_1,\ldots,c_{n-1}  \in \CC$, then 
$\hat{{\bf c}} = 1,c_1,\ldots,c_{n-1}$ is a state in $\CP$. A shift of 
$\hat{\bf c}$, say ${\bf u}:=c_j,c_{j+1},\ldots,c_{n-2},1,\ldots,c_j$, with $c_j=0$ for some $2 \leq j < n-1$, can then be uniquely determined in $\CP$. Its conjugate state $\hat{\bf u}=1,c_{j+1},\ldots,c_{n-2},1,\ldots,c_j$ is therefore in a CCR cycle $\CC'$. We next consider a shift of $\hat{\bf u}$, say ${\bf w} := 0, \bar{c}_1, \bar{c}_2,\ldots,\bar{c}_j,c_{j+1},\ldots,c_{n-2}, 1$, which must also be in $\CC'$. Because $\bar{c}_j=1$, we know that ${\bf w}$ is lexicographically less than the cycle representative ${\bf x}_{\mathbf{c}}=\bar{c}_1,\bar{c}_2,\ldots,\bar{c}_{n-2},0,c_1$ of $\CC$. Hence  $\CC'{ \prec}_{\rm lex} \CC$. In this case, we order $\CC' \prec \CP  \prec \CC$.  

We can now define the order among all cycles of the PRR. We note that each PCR cycle  $\CP \neq (1^{n-1})$ contains a uniquely determined state ${\bf c} = {\bf x}_{\mathbf{c}}$. Hence, there exists a unique CCR cycle $\CC$ containing $\hat{{\bf c}}$ such that $\CP$ is adjacent to $\CC$ and $\CC \prec \CP$. 

For $i$ from $1$ to $t$, we collect all PCR cycles $\CP_{i,1},\ldots,\CP_{i,s_i}$ that are adjacent to $\CC_i$. Obviously, $\CC_i\prec\CP_{i,j}$. Moreover, when $\CC_{i+1}$ exists, we fix $\CP_{i,j}\prec\CC_{i+1}$. The order on $\CP_{i,1},\ldots,\CP_{i,s_i}$ can be chosen arbitrarily. The cycles produced by the PRR are finally ordered as  
\begin{multline}
(1^{n-1}) \prec \CC_1\prec \mbox{(all PCR cycles adjacent with } \CC_1 \mbox{ in any order)} \prec \CC_2 \prec \\
\mbox{(all PCR cycles adjacent with} \CC_2 \mbox{ in any order)} \prec\CC_3
\prec \cdots \prec \\
\CC_t \prec \mbox{(all PCR cycles adjacent with } \CC_t \mbox{ in any order)}.
\end{multline}

The above procedure ensures that $A_{\Upsilon}$ is indeed a critical set. Each conjugate pair $({\bf c}, \hat{{\bf c}})$ contributes ${\bf c}$ and $\hat{{\bf c}}$ to $A_{\Upsilon}$. Again, either ${\bf c}$ or $\hat{{\bf c}}$ belongs to the same cycle  as ${\bf x}_{\mathbf{c}}$ does. Without loss of generality, we assume ${\bf c}$ and 
${\bf x}_{\mathbf{c}}$ belong to the same cycle. Then the cycle containing both 
${\bf x}_{\mathbf{c}}$ and ${\bf c}$ is greater than the cycle containing $\hat{\bf c}$ in our newly defined order. By Theorem \ref{thm:old}, the successor rule generates a de Bruijn sequence.
\end{proof}

We observe from the proof of Theorem~\ref{thm:rho-2} that ordering all cycles in a line is not strictly neccesary. We can in fact still be able to generate a de Bruijn sequence as long as all pairs of cycles containing the conjugate pairs in the set are comparable so that we can find a critical set of states that form the conjugate pairs. This means that we can slightly generalize Theorem~\ref{thm:old} by relaxing the condition that all cycles must be ordered. A typical rooted tree can be constructed following the successor rule $\Upsilon$ in the proof of Theorem~\ref{thm:rho-2}. We demonstrate it in  Figure~\ref{fig:f1}. 

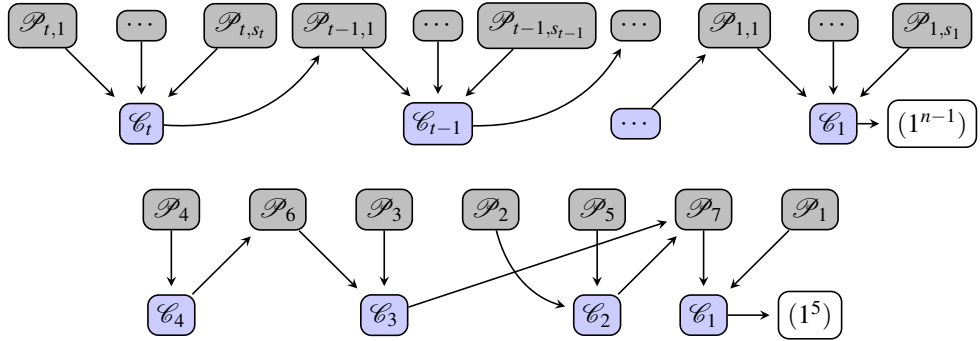
\begin{figure}
	\centering
	\begin{tikzpicture}
		[
		> = stealth,
		shorten > = 3pt,
		auto,
		node distance = 1.3cm,
		semithick
		]
		
		\tikzstyle{every state}=
		\node[rectangle,fill=white,draw,rounded corners,minimum size = 4mm]
		\node[state,fill=lightgray] (1) {$\CP_{t,1}$};
		\node[state,fill=lightgray] (2) [right of=1] {$\cdots$};
		\node[state,fill=lightgray] (3) [right of=2] {$\CP_{t,s_t}$};
		
		\node[state,fill=lightgray] (4) [right of=3]{$\CP_{t-1,1}$};
		\node[state,fill=lightgray] (5) [right of=4] {$\cdots$};
		\node[state,fill=lightgray] (6) [right of=5] {$\CP_{t-1,s_{t-1}}$};
		
		\node[state,fill=lightgray] (7) [right of=6] {$\cdots$};
		
		\node[state,fill=lightgray] (8) [right of=7]{$\CP_{1,1}$};
		\node[state,fill=lightgray] (9) [right of=8] {$\cdots$};
		\node[state,fill=lightgray] (10) [right of=9] {$\CP_{1,s_{1}}$};
		
		\node[state,fill=blue!20] (11) [below of=2] {$\CC_t$};
		\node[state,fill=blue!20] (12) [below of=5] {$\CC_{t-1}$};
		\node[state,fill=blue!20] (13) [below of=7] {$\cdots$};
		\node[state,fill=blue!20] (14) [below of=9] {$\CC_1$};
		\node[state] (15) [right of=14] {$(1^{n-1})$};
		
		\path[->] (1) edge (11);
		\path[->] (2) edge (11);
		\path[->] (3) edge (11);
		
		\path[->] (11) edge[bend right] (4);
		
		\path[->] (4) edge (12);
		\path[->] (5) edge (12);
		\path[->] (6) edge (12);
		
		\path[->] (12) edge[bend right] (7);
		\path[->] (13) edge (8);
		
		\path[->] (8) edge (14);
		\path[->] (9) edge (14);
		\path[->] (10) edge (14);
		
		\path[->] (14) edge (15);
	\end{tikzpicture}
	
	\vspace{0.5cm}
	
	\begin{tikzpicture}
		[
		> = stealth,
		shorten > = 3pt,
		auto,
		node distance = 1.4cm,
		semithick
		]
		
		\tikzstyle{every state}=
		\node[rectangle,fill=white,draw,rounded corners,minimum size = 4mm]
		\node[state,fill=lightgray] (1) {$\CP_4$};
		\node[state,fill=lightgray] (2) [right of=1] {$\CP_6$};
		\node[state,fill=lightgray] (3) [right of=2] {$\CP_3$};
		\node[state,fill=lightgray] (4) [right of=3]{$\CP_2$};
		\node[state,fill=lightgray] (5) [right of=4] {$\CP_5$};
		\node[state,fill=lightgray] (6) [right of=5] {$\CP_7$};
		\node[state,fill=lightgray] (7) [right of=6] {$\CP_1$};
		
		\node[state,fill=blue!20] (8) [below of=1] {$\CC_4$};
		\node[state,fill=blue!20] (9) [below of=3] {$\CC_3$};
		\node[state,fill=blue!20] (10) [below of=5] {$\CC_2$};
		\node[state,fill=blue!20] (11) [below of=6] {$\CC_1$};
		
		\node[state] (12) [right of=11] {$(1^{5})$};
		
		\path[->] (1) edge (8);
		\path[->] (2) edge (9);
		\path[->] (3) edge (9);
		
		\path[->] (4) edge[bend right] (10);
		\path[->] (5) edge (10);
		
		\path[->] (8) edge(2);
		
		\path[->] (9) edge (6);
		
		\path[->] (10) edge (6);
		
		\path[->] (6) edge (11);
		\path[->] (7) edge (11);
		
		\path[->] (11) edge (12);
		
	\end{tikzpicture}
	
\caption{{{\bf Above}: A typical rooted tree based on the successor rule $\Upsilon$, with $\CC_1=(0^{n-1}1^{n-1})$ and $\CP_{1,s_{1}}=(0^{n-1})$. Letting $t:=Z_{n-1}^*$, the CCR cyles are arranged in decreasing lexicographic order $\CC_t { \succ}_{\rm lex} \CC_{t-1} { \succ}_{\rm lex} \cdots { \succ}_{\rm lex} \CC_1$ from left to right. It may be the case that there are more than one CCR cycles, say $\CC_i$ and $\CC_j$ with $1 \leq i \neq j \leq t$, each having a directed edge to a common PCR cycle. {\bf Below}: A rooted tree when $n=6$, using the cycles specified in Example~\ref{ex:PCCR}. In each PCR cycle $\CP \neq (11111)$, its cycle representative is chosen as the uniquely determined state ${\bf x}_{\mathbf{c}}$, which happens to be the state ${\bf c}$ itself. Based on the respective cycle representatives of the CCR cycles, we use as our ${\bf c} \in A_{\Upsilon}$ the state $011111 \in \CC_1$, $011101 \in \CC_2$, $011011 \in \CC_3$, and $010101 \in \CC_4$.}}\label{fig:f1}
\end{figure}

By Theorem~\ref{thm:rho-2}, distinct ways of determining the unique state whose first bit is $0$ in any PCR cycle $\CP$ lead to different successor rules, generating inequivalent de Bruijn sequences. Similar with the operator $\Lambda$, we define an operator $\Theta$ that fixes $1^{n-1}$ and $01^{n-2}$. In all other cases, if $\mathbf{u} :=c_1, c_2,\ldots,c_{n-1}$ and $i$ is the least index for which $c_i=0$, with $2 \leq i < n$, then
\[
\Theta \, \mathbf{u} = c_{i},\ldots,c_{n-1},c_1,\ldots,c_{i-1} \mbox{, where } \Theta^0 \, \mathbf{u} := \mathbf{u} \mbox{ and } \Theta^r \, \mathbf{u} := \Theta^{r-1} \, (\Theta \, \mathbf{u}).
\]

The next two theorems construct numerous explicit successor rules based on Theorem \ref{thm:rho-2}. Their respective proofs follow the same line of argument as those of Theorems \ref{prop:1} and \ref{prop:2} and are omitted here for brevity.

\begin{theorem}\label{prop:3}
Let positive integer $n$, $t$, and $k_1, \ldots, k_t$ be such that
\[
n \geq 3, \quad 2 \leq t \leq n-1, \quad 1 = k_1 < k_2 < \ldots < k_{t} =n, \quad \mbox{and } k_{t-1} < n-1.
\]
For each state $\mathbf{c} = c_0,c_1,\ldots,c_{n-1}$, let $\mathbf{y}_{\mathbf{c}} := \bar{c}_1, \bar{c}_2, \ldots, \bar{c}_{n-2}, 0$ and $\mathbf{w}_{\mathbf{c}} := 0,c_1,\ldots,c_{n-2}$.
The successor rule
\begin{equation}\label{new5}
\Upsilon_1(\mathbf{c}) =
\begin{cases}
\overline{c_{0} + c_{1} + c_{n - 1}}  & \mbox{ if } c_{n-1}=1 \mbox{ and } \mathbf{y}_{\mathbf{c}} \mbox{ is a co-necklace}, \\
\overline{c_{0} + c_{1} + c_{n - 1}} & \mbox{ if } c_{n-1}=0 \mbox{ and there is an } i \mbox{ such that}\\
&\quad k_i \leq {\rm wt}(\bar{\mathbf{w}}_{\mathbf{c}}) < k_{i+1} \mbox{ and } \Theta^{k_i-1} \, \mathbf{w}_{\mathbf{c}} \mbox{ is a necklace},\\
c_{0} + c_{1} + c_{n - 1}    & \mbox{otherwise,}
\end{cases}
\end{equation}
generates $2^{n-3}$ inequivalent de Bruijn sequences of order $n$ by taking all possible parameters.
\end{theorem}

\begin{theorem}\label{prop:4}
For each $\mathbf{c} = c_0,c_1,\ldots,c_{n-1}$, let $\mathbf{y}_{\mathbf{c}} :=\bar{c}_1,\bar{c}_2,\ldots,\bar{c}_{n-2},0$ and $\mathbf{w}_{\mathbf{c}} := 0,c_1,\ldots,c_{n-2}$. Let $k$ be a nonnegative integer. The successor rule
\begin{equation}\label{new6}
\Upsilon_2(\mathbf{c}) =
\begin{cases}
\overline{c_{0} + c_{1} + c_{n - 1}}  & \mbox{ if } c_{n-1}=1 \mbox{ and } \mathbf{y}_{\mathbf{c}} \mbox{ is a co-necklace},\\
\overline{c_{0} + c_{1} + c_{n - 1}}  & \mbox{ if } c_{n-1}=0 \mbox{ and } \Theta^{k}\,\mathbf{w}_{\mathbf{c}} \mbox{ is a necklace}, \\
c_{0} + c_{1} + c_{n - 1}     & \mbox{otherwise,}
\end{cases}
\end{equation}
generates $\lcm(1,2,\ldots,n-2)$ inequivalent de Bruijn sequences of order $n$.
\end{theorem}

\begin{example}
Let us consider successor rules in the class $\Upsilon$ for $n=6$. Using Theorem~\ref{prop:3}, we obtain $8$ distinct successor rules, resulting in $8$ distinct de Bruijn sequences. Theorem~\ref{prop:4} gives us $12$ distinct de Bruijn sequences. In Table \ref{table:3}, the sequences in entries $1$, $2$, $3$, and $8$ on Theorem \ref{prop:3} are the same as the sequences in entries $2$, $3$, $4$, and $1$ on Theorem \ref{prop:4}, respectively. Table~\ref{table:3} contains $16$ distinct de Bruijn sequences in total.

\begin{table*}[h!]
\caption{Inequivalent de Bruijn sequences constructed based on Theorems~\ref{prop:3} and~\ref{prop:4} with $n=6$.}
\label{table:3}
\renewcommand{\arraystretch}{1.2}
\centering
{\footnotesize
	\begin{tabular}{ccc}
		\hline
		Entry & $\{k_1,k_2,\ldots,k_t\}$ & The resulting de Bruijn sequence based on Theorem~\ref{prop:3}\\
			\hline
			$1$   & $\{1,6\}$       & $(0000001111110100010000101110011101101010010101100100110001101111)$\\
			$2$   & $\{1,2,6\}$     & $(0000001111110101011010001011101100101001001101111001110001100001)$\\
			$3$   & $\{1,3,6\}$     & $(0000001111110100101000010001011100111011010101100011001001101111)$\\
			$4$   & $\{1,4,6\}$     & $(0000001111110100010111001110110101001010110010000100110001101111)$\\
			$5$   & $\{1,2,3,6\}$   & $(0000001111110101011010010100001000101110110001100100110111100111)$\\
			$6$   & $\{1,2,4,6\}$  & $(0000001111110101011010001011101100101001000010011011110011100011)$\\
			$7$   & $\{1,3,4,6\}$   & $(0000001111110100101000101110011101101010110001100100001001101111)$\\
			$8$   & $\{1,2,3,4,6\}$ & $(0000001111110101011010010100010111011000110010000100110111100111)$\\
			\hline
			\hline
			Entry & $k$  & The resulting de Bruijn sequence based on Theorem~\ref{prop:4}\\
			\hline	
			$1$  & $0$ & $(0000001111110101011010010100010111011000110010000100110111100111)$ \\
			$2$  & $1$  & $(0000001111110100010000101110011101101010010101100100110001101111)$ \\
			$3$  & $2$  & $(0000001111110101011010001011101100101001001101111001110001100001)$ \\
			$4$  & $3$  & $(0000001111110100101000010001011100111011010101100011001001101111)$ \\
			$5$  & $4$  & $(0000001111110101001010110100010111011001000010011000110111100111)$ \\
			$6$  & $5$  & $(0000001111110100010000101110011101101010110010100100110111100011)$ \\
			$7$  & $6$  & $(0000001111110101011010010100010111011000110010011011110011100001)$ \\
			$8$  & $7$  & $(0000001111110100001000101110011101101010010101100100110001101111)$ \\
			$9$  & $8$  & $(0000001111110101011010001011101100101001000010011011110011100011)$ \\
			$10$  & $9$  & $(0000001111110100101000100001011100111011010101100011001001101111)$ \\
			$11$  & $10$  & $(0000001111110101001010110100010111011001001100011011110011100001)$ \\
			$12$  & $11$ & $(0000001111110100001000101110011101101010110010100100110111100011)$ \\
			\hline
		\end{tabular}
}
\end{table*}
\end{example}

There are a number of alternatives to determine a unique state in a PCR cycle $\CP$ whose first bit is $0$ that will result in valid new successor rules for de Bruijn sequences. We omit the details here since Theorems~\ref{prop:3} and~\ref{prop:4} have already highlighted some of these possibilities.

\section{Complexity}\label{con}

We end by considering the complexity of the successor rules constructed in this paper. It is clear that the space complexity is $O(n)$. In the cycle structure of the PRR of order $n$, checking the cycle representative of a CCR cycle $\CC$ is equivalent to checking the co-necklace in a cycle generated by the CCR of order $n-1$. To pinpoint a unique state in a PCR cycle $\CP$ is equivalent to checking whether a state is a necklace in a cycle generated by the PCR of order $n-1$ after simple left shifts. The latter can be done in $O(n)$ time as was established in~\cite{Gabric18}. Thus, each successor rule in the two classes $\Lambda$ and $\Upsilon$ requires time and space complexities $O(n)$ to generate the next bit of a de Bruijn sequence of order $n$ from a given $n$-stage state.

Based on the special LFSR whose characteristic polynomial is $f_{\rm PRR}(x) = x^{n} + x^{n-1} + x + 1 = (x^{n-1}+1) (x+1)$ for $n \geq 3$, we have come up with two generic classes of successor rules. Each class contains numerous distinct successor rules, yielding mostly pairwise inequivalent de Bruijn sequences. The resulting family is of size $O(2^{n-3})$. The time and space complexities to generate the next bit in each of the instances are both $O(n)$. In the appendix we supply a basic implementation in {\tt C} that produces all of the de Bruijn sequences based on Theorems~\ref{prop:1}, \ref{prop:2}, \ref{prop:3}, and \ref{prop:4}.

The route that we propose here can be particularly useful to analyse the suitability of an arbitrary FSR whose cycles have small periods. Identifying more classes of suitable FSRs that efficiently produce larger families of de Bruijn sequences via successor rules is an interesting direction to investigate. Adding specific desirable properties for the resulting sequences could be an intriguing challenge to explore.

\section*{Acknowledgements}

We thank the reviewers for making us aware of an MSc thesis \cite{SalaTh}, a recent preprint \cite{Salasame}, and a webpage \cite{Sawada-web}. We also thank their helpful suggestions to clarify our contribution and improve the presentation of this paper.

The work of Z.~Chang is supported by the National Natural Science Foundation of China under Grant 61772476. Nanyang Technological University Grant Number M4080456 supports the work of M.~F.~Ezerman. The research of Q.~Wang is partially support by the Natural Sciences and Engineering Research Council of Canada (RGPIN-2017-06410).

\section*{Appendix: Source Code}

A basic implementation in {\tt C} is supplied for the interested readers.

\begin{lstlisting}[language=C]
/*
This is a basic implementation of the manuscript 
An Efficiently Generated Family of Binary de Bruijn Sequences
in submission to Discrete Mathematics
This v1.1 is dated October 1, 2020
*/

#include<stdio.h>
#include<math.h>
#include <stdlib.h>
int *ki_list,t,**all_ki,ki_case=0,lcmnum=0;

void lcm(int ar[], int size){
    int i =0;
    lcmnum = ar[0];
 
    while (1) {
      
      for (i = 0; i < size; i++) {
        if(lcmnum % ar[i]) break;
      }
      if( i == size) break;
      lcmnum++;
    }
}

int hamming_weight(int sequence[], int n) {
	int weight = 0;
	for (int j=0; j<n; j++){
		if (sequence[j] == 1) weight++;
	}
	return weight;
}

int if_necklace(int sequence[], int size){
	int p = 1;
	for (int j = 1; j<size; j++){
		if (sequence[j] < sequence[j-p]) return 0;
		if (sequence[j] > sequence[j-p]) p = j+1;
	}
	if (size % p != 0) return 0;
	return 1;
}

int shift(int sequence[], int size, int ki, int shift_case){
	int states[3*size], state[size];
	int shift_order = 0, position = size;
	for (int j=0; j<size; j++){
		states[j] = sequence[j];
		states[j+size] = sequence[j];
		states[j+2*size] = sequence[j];
	}
	while(shift_order < ki){
		if (states[++position] ==1) shift_order++;
	}
	if (shift_case == 1){
		for (int j=0; j<size; j++){
			state[j] = states[j+position+1-size];
		}
		if (if_necklace(state, size)) return 1;
		return 0;
	}
	if (shift_case == 2){
		for (int j=0; j<size; j++){
			state[j] = 1-states[j+position];
		}
		if (if_necklace(state, size)) return 1;
		return 0;
	}
}

int alg_theoremSeven(int sequence[], int n, int kicase){
	int i, j, u[n-1], u_co[2*(n-1)], number_list[n-1];
	if (kicase == pow(2, n-3)-1){
		number_list[0] = 1;
		number_list[1] = n;
		j = 1;
	}
	else{
		for(j=0;j<n-2;j++){
			if (all_ki[kicase][j]==0) break;
			number_list[j] = all_ki[kicase][j];
		}
		number_list[j] = n;
	}
	for (i=0;i<n-1;i++){
		u[i] = sequence[i+1];
		u_co[i] = sequence[i+1];
		u_co[i+n-1] = 1-sequence[i+1];
	}
	if (sequence[1]==0){
		if (if_necklace(u_co, 2*(n-1))) return (sequence[0]+sequence[1]+sequence[n-1]+1)%2;
	}
	else{
		int weight = hamming_weight(u, n-1);
		for (i=0;i<j;i++){
			if (number_list[i] <= weight && weight < number_list[i+1]){
				if (shift(u, n-1, number_list[i]-1, 1)) return (sequence[0]+sequence[1]+sequence[n-1]+1)%2;
			}
		}
	}
	return (sequence[0]+sequence[1]+sequence[n-1])%2;
}

int alg_theoremNine(int sequence[], int n, int k_number){
	int i, u[n-1], u_co[2*(n-1)];
	for (i=0;i<n-1;i++){
		u[i] = sequence[i+1];
		u_co[i] = sequence[i+1];
		u_co[i+n-1] = 1-sequence[i+1];
	}
	if (sequence[1]==0){
		if (if_necklace(u_co, 2*(n-1))) return (sequence[0]+sequence[1]+sequence[n-1]+1)%2;
	}
	else{
		if (shift(u, n-1, k_number % hamming_weight(u, n-1), 1)) return (sequence[0]+sequence[1]+sequence[n-1]+1)%2;
	}
	return (sequence[0]+sequence[1]+sequence[n-1])%2;
}

int alg_theoremEleven(int sequence[], int n, int kicase){
	int i, j, u[n-1], u_co[2*(n-1)], w_bar[n-1], number_list[n-1];
	if (kicase == pow(2, n-3)-1){
		number_list[0] = 1;
		number_list[1] = n;
		j = 1;
	}
	else{
		for(j=0;j<n-2;j++){
			if (all_ki[kicase][j]==0) break;
			number_list[j] = all_ki[kicase][j];
		}
		number_list[j] = n;
	}
	w_bar[0] = 1;
	u[n-2] = 0, u_co[n-2] = 0, u_co[2*n-3] = 1;
	for (i=0;i<n-2;i++){
		u[i] = 1-sequence[i+1];
		u_co[i] = 1-sequence[i+1];
		u_co[i+n-1] = sequence[i+1];
		w_bar[i+1] = 1-sequence[i+1];
	}
	if (sequence[n-1]==1){
		if (if_necklace(u_co, 2*(n-1))) return (sequence[0]+sequence[1]+sequence[n-1]+1)%2;
	} 
	else{
		int zero_number = hamming_weight(w_bar, n-1);
		for (i=0;i<j;i++){
			if (number_list[i] <= zero_number && zero_number < number_list[i+1]){
				if (shift(w_bar, n-1, number_list[i]-1, 2)) return (sequence[0]+sequence[1]+sequence[n-1]+1)%2;
			}
		}
	}
	return (sequence[0]+sequence[1]+sequence[n-1])%2;
}

int alg_theoremTwelve(int sequence[], int n, int k_number){
	int i, u[n-1], u_co[2*(n-1)], w_bar[n-1];
	w_bar[0] = 1;
	u[n-2] = 0, u_co[n-2] = 0, u_co[2*n-3] = 1;
	for (i=0;i<n-2;i++){
		u[i] = 1-sequence[i+1];
		u_co[i] = 1-sequence[i+1];
		u_co[i+n-1] = sequence[i+1];
		w_bar[i+1] = 1-sequence[i+1];
	}
	if (sequence[n-1]==1){
		if (if_necklace(u_co, 2*(n-1))) return (sequence[0]+sequence[1]+sequence[n-1]+1)%2;
	}
	else{
		if (shift(w_bar, n-1, k_number % hamming_weight(w_bar, n-1), 2)) return (sequence[0]+sequence[1]+sequence[n-1]+1)%2;
	}
	return (sequence[0]+sequence[1]+sequence[n-1])%2;
}

void comb(int m, int k){
    int i, j;
    for(i=m; i>=k; i--){
        ki_list[k-1] = i;
        if(k>1) comb(i-1, k-1);
        else{
        	all_ki[ki_case][0] = 1;
            for(j=0; j<t; j++) all_ki[ki_case][j+1] = ki_list[j]+1;
            ki_case++;
        }
    }
}

void DB(int alg_number, int n) {
	int i, j, new_bit, kicase, k_number, state[n];
	if (alg_number == 1 || alg_number == 3){
		for (kicase=0; kicase<pow(2, n-3); kicase++){
			for (i=0; i<n; i++) state[i] = 0;
			for(j=0;j<n-2;j++){
				if (all_ki[kicase][j]==0) break;
				printf("%d", all_ki[kicase][j]);
			}
			printf(" ");
			do {
				printf("%d", state[0]);
				switch(alg_number) {
					case 1: new_bit = alg_theoremSeven(state, n, kicase); break;
					case 3: new_bit = alg_theoremEleven(state, n, kicase); break;
					default: break;
			}
				for (i=0; i<n; i++) state[i] = state[i+1];
				state[n-1] = new_bit;
			} while (hamming_weight(state, n)>0);
			printf("\n");
		}
	}
	else{
		for (k_number=1; k_number<=lcmnum; k_number++){
			for (i=0; i<n; i++) state[i] = 0;
			printf("%d", k_number);
			printf(" ");
			do {
				printf("%d", state[0]);
				switch(alg_number) {
					case 2: new_bit = alg_theoremNine(state, n, k_number); break;
					case 4: new_bit = alg_theoremTwelve(state, n, k_number); break;
					default: break;
				}
				for (i=0; i<n; i++) state[i] = state[i+1];
				state[n-1] = new_bit;
			} while (hamming_weight(state, n)>0);
			printf("\n");
		}
	}
}

void main(){
	int n, k, alg_number, i;
	printf("Enter n:"); scanf("%d", &n);
	int max_number = pow(2, n-3);
	ki_list = (int*)calloc((n-3),sizeof(int));
	all_ki = (int**)calloc(max_number,sizeof(int*));
	for(i=0;i<max_number;i++) {
		all_ki[i] = (int*)calloc(n-2,sizeof(int));
	}
	for (i=n-3; i>0; i--){
		t = i;
	    comb(n-3,i);
	}
	all_ki[ki_case][0] = 1;
	int klist[n-2];
	for (i=0;i<n-2;i++){
		klist[i] = i+1;
	}
	lcm(klist, n-2);
	for (alg_number=1; alg_number<5; alg_number++){
		DB(alg_number, n); 
		printf("\n");
	}
	for(i=0;i<max_number;i++) free(all_ki[i]);
	free(all_ki), free(ki_list);
}
\end{lstlisting}

\end{document}